%% file: main.tex
\documentclass[a4paper]{scrartcl}

\usepackage{amsmath,amsthm,amsfonts,amssymb,mathrsfs}
\usepackage{fullpage}
\usepackage[hidelinks]{hyperref}
\usepackage{cleveref, thm-restate}
\usepackage{subcaption}
\usepackage{complexity}
\usepackage{wrapfig}
\usepackage{commath}
\usepackage[ruled,nofillcomment]{algorithm2e}
\usepackage{mathtools}
\usepackage{todonotes}
\usepackage{setspace}
\usepackage{enumitem}
\usepackage{multirow}
\usepackage{makecell}
\usepackage{thmtools, thm-restate}
\usepackage{xspace}
\usepackage{tikz}
\usetikzlibrary{calc}
\usepackage{tkz-euclide}
\pgfdeclarelayer{background}
\pgfdeclarelayer{foreground}
\pgfsetlayers{background,main,foreground}

\newtheorem{theorem}{Theorem}

\newtheorem{definition}[theorem]{Definition}

\newtheorem{lemma}[theorem]{Lemma}

\DeclareMathOperator{\cost}{cost}

\DeclareMathOperator{\old}{old}

\let\poly\relax
\DeclareMathOperator{\poly}{poly}

\newcommand{\eps}{\varepsilon}
\newcommand{\epsp}{\varepsilon'}
\newcommand{\ebalanced}{$\eps$-balanced\xspace}
\newcommand{\opt}{\text{opt}}

\newcommand{\mball}[2][]{\ensuremath{\text{MB}_{#1}{\left( #2 \right)}}}
\newcommand{\ball}[3][]{\ensuremath{\text{B}_{#1}{\bigl( #2, \, #3 \bigr)}}}

\newcommand{\col}{col}

\DeclarePairedDelimiter\ceil{\lceil}{\rceil}

\newcommand{\orcidID}[1]{\href{https://orcid.org/#1}{\hspace*{0.1cm}\includegraphics{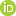}}}

\title{Approximating Fair \texorpdfstring{$k$}{k}-Min-Sum-Radii in Euclidean Space\thanks{Funded by \emph{Deutsche Forschungsgemeinschaft} (DFG, German Research Foundation) – Project 456558332.}}

\author{Lukas Drexler\orcidID{0000-0001-9395-6711} \and
Annika Hennes\orcidID{0000-0001-9109-3107} \and
Abhiruk Lahiri\orcidID{0009-0008-7556-3445} \and
Melanie Schmidt\orcidID{0000-0003-4856-3905} \and
Julian Wargalla\orcidID{0000-0003-4583-7288}}
\begin{document}

\maketitle

\begin{abstract}
The $k$-center problem is a classical clustering problem in which one is asked to find a partitioning of a point set $P$ into $k$ clusters such that the maximum radius of any cluster is minimized. It is well-studied. But what if we add up the radii of the clusters instead of only considering the cluster with maximum radius? This natural variant is called the $k$-min-sum-radii problem. It has become the subject of more and more interest in recent years, inspiring the development of approximation algorithms for the $k$-min-sum-radii problem in its plain version as well as in constrained settings.

We study the problem for Euclidean spaces $\mathbb{R}^d$ of arbitrary dimension but assume the number $k$ of clusters to be constant. In this case, a $\PTAS$ for the problem is known (see Bandyapadhyay, Lochet and Saurabh~\cite{BandyapadhyayLS23}). 
Our aim is to extend the knowledge base for $k$-min-sum-radii to the domain of \emph{fair} clustering. We study several group fairness constraints, such as the one introduced by Chierichetti et al.~\cite{chierichetti2017fair}. In this model, input points have an additional attribute (e.g., colors such as red and blue), and clusters have to preserve the ratio between different attribute values (e.g., have the same fraction of red and blue points as the ground set). Different variants of this general idea have been studied in the literature. To the best of our knowledge, no approximative results for the fair $k$-min-sum-radii problem are known, despite the immense amount of work on the related fair $k$-center problem.

We propose a $\PTAS$ for the fair $k$-min-sum-radii problem in Euclidean spaces of arbitrary dimension for the case of constant $k$. To the best of our knowledge, this is the first $\PTAS$ for the problem. It works for different notions of group fairness.
\end{abstract}

\newpage \clearpage
\setcounter{page}{1}
\section{Introduction}

The \emph{$k$-min-sum-radii} problem (\textsf{$k$-MSR} for short) is a clustering problem that resembles two well-known problems, namely the \emph{$k$-center} problem and the \emph{$k$-median} problem. 
Given a set of points $P$ and a number $k$, these problems ask for a set $C$ of $k$ cluster centers and evaluate it according to the distances $d_{\min}(x) = \min \{ d(c,x) \mid c \in C \}$ between points $x$ and their closest center in $C$.
The $k$-center objective $\max\{d_{\min}(x) \mid x \in P \}$ focuses on the radii of the resulting clusters, while the $k$-median objective $\sum_{x \in P} d_{\min}(x)$ sums up all individual point's costs. 
In the latter case, large individual costs can average out and so in scenarios where we really want to restrict the maximum cost an individual point can induce, $k$-center is the better choice. However, looking only for a single maximum distance completely ignores the fine-tuning of up to $k-1$ smaller clusters. 

The \emph{$k$-min-sum-radii} problem goes an intermediate way: 
It looks for $k$ centers $c_1,\ldots,c_k$ and corresponding clusters $C_1,\ldots, C_k$ and sums up the radii, i.e., the objective is to minimize $\sum_{i=1}^k \max_{x \in C_i} d(x,c_i)$ (or in other words: to minimize the \emph{average} radius). This objective allows for the fine-tuning of all clusters while still maintaining that the maximum cost of individual points is reasonably bounded (although it may be higher than for $k$-center by a factor of $k$). Another variation, known as the \emph{$k$-min-sum-diameter} problem, aims to minimize the sum of the diameters of the clusters.
The $k$-min-sum-radii problem has a close connection with the base station placement problem arising in wireless network design~\cite{Lev-TovP05}, where the objective is to minimize the energy required for wireless transmission which is proportional to the sum of the radii of coverage of the base stations. The mathematical model of this problem translates to the \emph{minimum sum radii cover problem} where we have a set of client locations and a set of server locations. The objective is to cover the set of clients with a set of balls whose centers are located at a subset of server locations such that the sum of the radii of the balls is minimized. 

There has been great interest in designing good approximation algorithms for the $k$-min-sum-radii problem. 
Charikar and Panigrahy~\cite{CP04} give an $O(1)$-approximation for the metric $k$-min-sum-radii problem (and the $k$-min-sum-diameter problem) based on the primal-dual framework by Jain and Vazirani for $k$-median. 
It was recently refined by Friggstad and Jamshidian~\cite{FJ22} to obtain a $3.389$-approximation for $k$-min-sum-radii which is currently the best-known approximation factor for the general case.
For constrained $k$-min-sum-radii, lower bounds, outliers and capacities have been studied. Ahmadian and Swamy~\cite{AS16} built upon~\cite{CP04} to obtain a $3.83$-approximation for the non-uniformly lower bounded $k$-min-sum-radii problem. They also give a $(12.365 + O(\eps))$-approximation for \textsf{$k$-MSR} with outliers that runs in time $n^{O(1/\eps)}$. 
Inamdar and Varadarajan~\cite{IV20} derive a $28$-approximation for the uniformly capacitated $k$-min-sum-radii problem, but this algorithm is an $\FPT$ approximation algorithm with running time $O(2^{O(k^2)}\cdot n^{O(1)})$. Bandyapadhyay, Lochet, and Saurabh~\cite{BandyapadhyayLS23} also give an $\FPT$-approximation: They develop a $(4+\eps)$-approximation algorithm with $2^{O(k\log (k/\eps))}\cdot n^3$ running time for \textsf{$k$-MSR} with uniform capacities and a $(15 + \eps)$-approximation algorithm for \textsf{$k$-MSR} with non-uniform capacities that runs in time $2^{O(k^2 \log k)}\cdot n^3$. 

In the Euclidean case, it is possible to obtain better results.
In the plane, every cluster in the optimal min-sum radii clustering lies inside some convex polygon drawn from the solution centers that partition the plane into k disjoint convex regions. The dual of that partition is an internally triangulated planar graph. Capolyleas et al. \cite{CRW1991} use this fact to enumerate over $O(n^{6k})$ possible solutions to solve the problem exactly for $d=2$. Gibson et al.~\cite{GibsonKKPV12} also give an exact algorithm for the Euclidean $k$-min-sum-radii problem in the plane. Their algorithm is based on an involved dynamic programming approach and has a running time of $O(n^{881})$ for $d=2$. 
Bandyapadhyay, Lochet and Saurabh~\cite{BandyapadhyayLS23} give a randomized algorithm with running time $2^{O((k/\eps^2)\log k)}\cdot dn^3$ which outputs a $(1+\eps)$-approximation with high probability. Their algorithm can handle capacitated $k$-min-sum-radii but allows the capacities to be violated by at most an $\eps$-fraction. They also present a $\PTAS$ for \textsf{$k$-MSR} with strict capacities for both constant $k$ and constant $d$ with running time  $2^{O(kd\log(\frac{k}{\eps}))}n^3$.

In this paper, we advance the active research on fairness in clustering (see \cite{caton2020fairness, chhabra2021overview} for surveys on the topic) 
and tackle the problem of $k$-min-sum-radii under a variety of different group fairness notions. 
These notions assume that the data points belong to different protected groups, represented by different colors. We will denote the set of colors by $\mathcal{H}$. For $X\subseteq P$ and $h\in \mathcal{H}$, let $\col_{h}(X) \subseteq X$ denote the subset of points within $X$ that carry color $h$. Now the notion of \emph{exact fairness} requires that in every cluster the proportion of points of a certain color is the same as their proportion within the complete point set, i.e., a clustering $\mathscr{C}$ fulfills {\em exact fairness} if 
$ \frac{|\col_{h}(C)|}{|C|} = \frac{|\col_h(P)|}{|P|} $
for every color $h\in \mathcal{H}$ and cluster $C\in \mathscr{C}$. This notion is for example defined in~\cite{rosner2018privacy}. Our method can handle exact fairness but also other notions as it is indeed capable to handle the more general class of \textit{mergeable constraints}.  
A clustering constraint is \emph{mergeable} if the union $C \cup C'$ of any possible pair of clusters $C, C'$ satisfying the constraint does itself satisfy the constraint (cf.~\cite{Arutyunova021}). In other words, merging clusters does not destroy the property of satisfying the constraint.

Important examples of mergeable constraints are (a) several \emph{fairness constraints} (see the appendix for a list),
(b) {\em lower bound constraints} that require every cluster to contain at least a certain fixed number of points, and (c) \emph{outliers} (see the end of the paper) in which a fixed number $z$ of points can be ignored by any clustering (one can model this as a kind-of-mergeable constraint by viewing it as a $(k+z)$-clustering with the constraint that at most $k$ clusters contain more than one point and the rest is singleton clusters). On the other hand, upper bounds on the cardinality of clusters (capacities) are not mergeable because merging clusters may violate the capacity constraint.

\input{drawing}

To the best of our knowledge, no results for $k$-min-sum-radii with fairness constraints are known, neither in the Euclidean setting ($\PTAS)$ nor in the metric setting (constant factor approximation) despite the huge amount of work on fair $k$-center and fair $k$-median (cf. ~\cite{bera2019fair, bercea2018cost, bohm2021algorithms, bohm2020fair,  chierichetti2017fair, harb2020kfc, rosner2018privacy})\footnote{The result by~\cite{CRW1991} might be extendable to the setting of fair $k$-min-sum-radii to obtain an exact algorithm for constant $k$ and $d=2$, but an in-depth analysis would be required to verify this idea.}. 
The reason for this may be that the $k$-min-sum-radii problem can behave quite counter-intuitively and has properties unlike both the $k$-center and $k$-median problem. One such property is that a $k$-min-sum-radii solution may actually cost \emph{more} if we open more centers (see Figure~\ref{fig:counterintuitive}, left side) which cannot happen for (plain) $k$-center or $k$-median. This is a problem for the design of fair clustering algorithms because for $k$-center and $k$-median, these are built by computing fair micro-clusters (also called fairlets) first and then assembling the final fair clustering from the micro-clusters (cf.~\cite{chierichetti2017fair}). Another uncommon property of $k$-min-sum-radii is that even without any constraints, assigning points to centers does not have the locality property: It may be
beneficial to assign to a further away center (see Figure~\ref{fig:counterintuitive}, right side). This has been observed for other clustering objectives when side constraints enter the picture, but for $k$-min-sum-radii, it already happens without any constraints.

\paragraph*{Our Result}

We present a simple $\PTAS$ for the Euclidean \textsf{$k$-MSR} problem with mergeable constraints that works for constant $k$ and arbitrary dimension $d$. 
In particular, to the best of our knowledge, we provide the first approximate results for fair \emph{$k$-MSR}.

\begin{restatable}{theorem}{maintheorem}
    \label{thm:result}
    For every  $0< \eps < 1/2$, there exists an algorithm that computes a $(1+\eps)$-approximation for $k$-min-sum-radii with mergeable constraints in time $d \cdot \poly(n) \cdot f(k, \eps)$, if the corresponding constrained $k$-center problem has a constant-factor polynomial time approximation algorithm.
    If no such $k$-center approximation exists, the running time increases to $d \cdot n^{\poly(k, 1/\eps)} \cdot f(k, \eps)$ (see Theorem~\ref{thm:result_2}).
\end{restatable}

\subparagraph*{How We Obtain the PTAS.}
\label{subpar:PTAS}
Our algorithm is based on an idea by B\u{a}doiu, Har-Peled, and Indyk \cite{BadoiuHI02} who obtain a $\PTAS$ for $k$-center. Contrary to other $k$-center algorithms, the main idea of this algorithm -- to iteratively construct minimum enclosing balls around subsets of optimum clusters until all points are covered -- \emph{does} carry over to \textsf{$k$-MSR}. However, we need to resolve significant obstacles that are due to the more complex structure of \textsf{$k$-MSR}, as illustrated in Figure~\ref{fig:counterintuitive}. 
In an optimal clustering, points do not necessarily get assigned to their closest center, so we cannot derive a lower bound for the initial size of the growing balls in the same manner as B\u{a}doiu et al. This, however, is necessary to upper bound the running time of the algorithm.

Our approach to repair the analysis is mainly based on proving that there always exists a close-to-optimum \textsf{$k$-MSR} solution with a nice structure, as described in Section~\ref{sec:nicelystructuredsolutions}: (1) The Minimum Enclosing Balls (MEBs) around all clusters do not intersect, even if we enlarge all MEBs by some factor $\gamma$ that depends on $\eps$. We call such a solution $\gamma$-separated. (2) The ratio between the smallest and the largest radius in the solution is bounded by $\eps/k$. We call a solution with this property \emph{$\eps$-balanced}. 
Achieving (2) is straightforward, but establishing (1) and (2) simultaneously requires a bit more work. Since we establish (1) mainly by merging close clusters, this technique still works under mergeable constraints. 

After proving the existence of an approximately optimal solution that is sufficiently separated and balanced, we reconstruct this solution by adjusting the approach of B\u{a}doiu et al.~\cite{BadoiuHI02} appropriately.
To ensure an upper bound on the running time, we  have to extend their guessing oracle (that answers membership queries) to also provide approximate radii for all clusters. How this is done is outlined in Section~\ref{sec:guessingradii}.
With the oracle in place, the structure of our algorithm is as follows:

\begin{itemize}
    \item Initialize $S_i=\emptyset$ for $i=1,\ldots,k$ and $P'=P$
    \item Ask the oracle for radii $\widetilde{r_1}, \widetilde{r_2}, \ldots, \widetilde{r_k}$
    \item Repeat until $P' = \emptyset$:
    \begin{enumerate}
        \item Select an arbitrary point $p_i$ from $P'$
        \item Query the oracle for an index $j$ and add $p_i$ to $S_j$
        \item \label{item:starting-distance}
        If $|S_j|=1$: Remove all points from $P'$ that are within distance $\approx\hspace*{-0.05cm}\eps \widetilde{r_j}$ of $p_i$
        \item \label{item:sufficient-additional-distance}
        If $|S_j| > 1$: Compute the minimum enclosing ball of $S_j$, enlarge it by an appropriate\\ \hspace*{1.7cm} factor and remove all points in the resulting ball from $P'$
    \end{enumerate}
\end{itemize}

\input{example-arxiv-pdfs}
We show that the algorithm will stop after $f(k,\eps)$ iterations, resulting in a $\PTAS$ for constant $k$. Figure~\ref{fig:selection} shows an example run of the algorithm and also gives some details on what \lq $\approx \eps \widetilde{r_j}$\rq\ and \lq by an appropriate factor\rq\ mean. The respective constants are a result of the analysis and are discussed later.

\subparagraph*{Further Related Work.}
When $k$ is part of the input, the metric $k$-min-sum-radii problem is known to be $\NP$-hard, as shown in~\cite{ProiettiW06}.
The same paper gives an exact algorithm with running time~$O(n^{2k}/k!)$. Gibson et al.~\cite{GibsonKKPV10} provide a randomized algorithm for the metric $k$-min-sum-radii problem that runs in time $n^{O(\log n \log \Delta)}$ where $\Delta$ is the ratio between the largest and the smallest pairwise distance in the input and returns an optimal solution with high probability. They also show $\NP$-hardness even for shortest path metrics in weighted planar graphs and for metrics of (large enough) constant doubling dimension.
Bilò et al.~\cite{BiloCKK05} give a polynomial time algorithm for the problem when the input points are on a line. 

Behsaz and Salavatipour~\cite{BehsazS15} show a polynomial time exact algorithm for the $k$-min-sum-radii problem when the metric is induced by an unweighted graph and no cluster contains only one point. 

Based on the constant-factor approximation for $k$-min-sum-radii by Charikar and Panigrahy~\cite{CP04} mentioned in the introduction, Henzinger et al.~\cite{HenzingerLM17} develop a data structure to efficiently maintain an $O(1)$-approximate solution under changes in the input.

For the $k$-center problem with exact fairness constraints as described earlier, Bercea et al. \cite{bercea2018cost} give a 5-approximation.
Further, several {\em balance} notions have been proposed. The simplest case with only two colors was proposed by Chierichetti et al.~\cite{chierichetti2017fair}. It requires that the minimum ratio between different colors within any cluster meets a given lower bound. For its most general formulation, there exists a $14$-approximation for the $k$-center variant \cite{rosner2018privacy}.
The definition by Böhm et al.~\cite{bohm2020fair} allows more colors but is stricter in that it demands the portions of colors in a cluster to be of equal size. The authors show how the $k$-center problem under this fairness notion can be reduced to the unconstrained case while increasing the approximation factor by 2, leading to a polynomial-time $O(1)$-approximation. They also give an $O(n^{\poly(k/\eps)})$-time $(1+\eps)$-approximation. A more general notion by Bera et al.~\cite{bera2019fair} allows the number of cluster members of a certain color to lie in some color-dependent range. Harb and Shan \cite{harb2020kfc} give a 5-approximation for the $k$-center problem under this constraint.

\subparagraph*{Preliminaries.}
For a given center $c \in \mathbb{R}^d$ and radius $r \in \mathbb{R}_{\geq 0}$, define the \emph{ball of radius $r$ around $c$} to be $\ball{c}{r} = \{x \in \mathbb{R}^d \; | \; \|x - c\| \leq r\}$. 
We set $\cost\bigl(\ball{c}{r}\bigr) = r$. Let $X \subset \mathbb{R}^d$ be a set of points. We say that a ball $B$ \emph{encloses} $X$ if $X \subset B$. The ball with the smallest radius that encloses $X$ is called the \emph{minimum enclosing ball} (MEB) of $X$ and we denote it as $\mball{X}$.
The \emph{cost} of $X$ is defined as the cost of its minimum enclosing ball, $\cost(X) = \cost(\mball{X})$. 
A \emph{$k$-clustering} $\mathscr{C} = \{C_1, \ldots, C_k\}$ of a given finite set of points $P$ is a partitioning of $P$ into $k$ disjoint (possibly empty) sets. Its cost is the sum of the costs of all its individual clusters, i.e. $\cost(\mathscr{C}) = \sum_i \cost(C_i)$. 
Now we can define the {\em Euclidean $k$-min-sum-radii problem}:
Given a finite set of points $P$ in the $d$-dimensional Euclidean space $\mathbb{R}^d$ and a number $k\in \mathbb{N}$, find a $k$-clustering of $P$ with minimal cost. We can also formulate the problem in the following form: Find at most $k$ centers $c_1, \ldots, c_k \in \mathbb{R}^d $ and radii $r_1, \ldots, r_k \ge 0$ such that the union of balls $B(c_1, r_1) \cup \ldots \cup B(c_k, r_k)$ covers $P$ and the sum of the radii $\sum_{i=1}^k r_i$ is minimized. 

\section{\texorpdfstring{$k$}{k}-Min-Sum-Radii with Mergeable Constraints}
Algorithm~\ref{alg:selection} gives a detailed description of our method. Instead of an oracle, this pseudo code assumes that it is given a string $u \in \{1, \ldots, k\}^*$ as answers to membership queries and it is also given estimates for the radii $\widetilde{r_1}, \ldots, \widetilde{r_k}$.  Despite looking a bit more technical, the algorithm follows the plan outlined above: Iteratively construct balls to cover the point set. A ball shall always cover exactly one optimum cluster (or approximately optimal cluster). It starts out when the first point from that cluster is discovered. This point will be the center of a small starting ball. The starting radius is related to the true radius of the cluster (for which a good estimate has been provided). Whenever a point from a cluster is discovered, the MEB around all its discovered points is computed and the ball is increased to an enlarged version of that MEB. We are done when all points are covered by the balls.

\begin{algorithm}[h]
    \setstretch{1.2}
    \LinesNumbered
    \caption{\textsc{Selection}}
    \label{alg:selection}
    \SetKwInOut{Input}{Input}
    \SetKwInOut{Output}{Output}
    \BlankLine
    \Input{An ordered set $P \subseteq \mathbb{R}^d$, a number $k \in \mathbb{N}$, a string $u \in \{1, \ldots, k\}^*$, a set $\{\widetilde{r_1}, \ldots, \widetilde{r_k}\}$ of $k$ radii, a value $0 < \eps < 1$}
    \Output{Balls $B_1, \ldots, B_k \subseteq \mathbb{R}^d$, such that $P \subseteq \bigcup_j B_j$}
    \BlankLine
        $S_i \gets \emptyset$ for $i = 1, \ldots, k$\;
        $\gamma \gets 1 + \eps + 2\sqrt{\eps}$\;
        $X \gets \emptyset$\tcc*[r]{points that have been covered so far}
        \For{$i = 1, \ldots, |u|$}{
            $I = \{j \mid S_j = \{s_j\} \text{ is a singleton}\}$\;
            $R \gets \bigcup_{j \in I} \ball{s_j}{\frac{\eps}{1+\eps}\widetilde{r_j}}$ \label{line:epsilonball}\tcc*[r]{put small balls around singletons}
            Let $p_i$ be the point from $P \setminus (X \cup R)$ that is first in the order induced by $P$\; \label{line:select-point}
            $S_{u_i} \gets S_{u_i} \cup \{ p_i \}$\;
            \If{$|S_{u_i}| \geq 2$}
            {
                $\ball{c}{r} \gets (1 + \eps)$-approximation of $\mball{S_{u_i}}$\;
                $B_{u_i} \gets \ball{c}{\gamma r}$\;
                $X \gets X \cup (B_{u_i} \cap P)$\;\label{line:ball}
            }
        }
        \ForAll{$S_i = \{s_i\}$}{
            $B_i \gets \ball{s_i}{0}$
            ;
        }
        \Return $B_1, \ldots, B_k$\;
\end{algorithm}

To set up the analysis, we introduce a few more definitions. In the preliminaries, we have defined a clustering to be a partitioning of the underlying space. In the following section, however, it will be helpful to occasionally conceive clusterings as collections of balls that cover $P$. To avoid confusion, we term the latter coverings.

\begin{definition}
    We say that a set of balls $\{\ball{c_1}{r_1}, \ldots \ball{c_k}{r_k}\}$ forms a \emph{covering} of $P$, if $P \subseteq \bigcup_i \ball{c_i}{r_i}$. It is a \emph{disjoint} covering, if $\ball{c_i}{r_i} \cap \ball{c_j}{r_j} = \emptyset$ for all $i \neq j$.
\end{definition}

The relation between coverings and clusterings, as it concerns this paper, is straightforward. Every disjoint covering $B_1, \ldots, B_k$ of $P$ yields a unique corresponding clustering $\mathscr{C} = \{B_1 \cap P, \ldots, B_k \cap P\}$ of $P$. And, conversely, every $k$-clustering $\mathscr{C} = \{C_1, \ldots, C_k\}$ yields a corresponding covering $\{\mball{C_1}, \ldots, \mball{C_k}\}$. 

We increase the size of the balls computed during Algorithm~\ref{alg:selection} by some multiplicative factor to ensure that they grow reasonably fast, so we want them to be not only disjoint but actually separated by some positive amount.

\begin{definition}
    Let $\gamma \geq 1$. Two balls $\ball{c_1}{r_1}$ and $\ball{c_2}{r_2}$ are said to be \emph{$\gamma$-separated} if $\ball{c_1}{\gamma r_1} \cap \ball{c_2}{\gamma r_2} = \emptyset$. A covering is \emph{$\gamma$-separated} if all its balls are pairwise $\gamma$-separated.
\end{definition}

\subsection{The Main Algorithm and The Main Lemma}

The following constitutes the main technical lemma. 
Claim 1(a) is mainly an observation: For some sequence of oracle guesses (i.e., for some $u \in \{1,\ldots,k\}^{\ast}$), we always guess the cluster correctly and thus $S_i$ is always a subset of some optimum cluster.  In fact, in the end each $S_i$ can be viewed as a compact approximation in the following sense: it is a small set whose MEB covers almost all of its corresponding optimum cluster. Claims 1(b) and 1(c) state that during the algorithm, the balls that are placed around the sampled points are always disjoint. This is important for the rest of the analysis. Then Claim (2) and (3) are the core part of the original analysis by \cite{BadoiuHI02}: Whenever we add a point, the ball for the cluster grows by an appropriate factor, and once a certain threshold of points has been reached, the ball covers the true cluster we are looking for. 

\begin{lemma}
    \label{lem:construct-u}
    Let $\mathscr{B} = \{\ball{c^*_1}{r^*_1}, \ldots, \ball{c^*_k}{r^*_k}\}$ be an arbitrary covering of $P$ and $\widetilde{r_1}, \ldots, \widetilde{r_k}$ a set of radii such that $r^*_i \leq \widetilde{r_i} \leq (1 + \eps)r^*_i$ for all $i \in \{1, \ldots, k\}$. If $\mathscr{B}$ is $(1 + \eps)\gamma$-separated, with $\gamma \geq 1 + \eps + 2\sqrt{\eps}$, then there is an element $u \in \{1, \ldots, k\}^*$, such that

    \begin{enumerate}
        \item At each stage of $\textsc{Selection}(P, k, u, \{\widetilde{r_1}, \ldots, \widetilde{r_k}\}, \eps)$, the following holds for all $i$:
        \begin{enumerate}[label=(\alph*)]
            \item $S_i \subseteq \ball{c^*_i}{r^*_i}$,
            \label{item:subset}
            \item $\ball{c^*_j}{(1 + \eps)\gamma r^*_j} \cap \ball{s_i}{\frac{\eps}{1+\eps}\widetilde{r_i}} = \emptyset$ for all $j \neq i$ whenever $S_i = \{s_i\}$ is a singleton,
            \label{item:singleton}
            \item $\ball{c^*_j}{(1 + \eps)\gamma r^*_j} \cap B_i = \emptyset$ for all $j \neq i$.
            \label{item:balls-disjoint}
        \end{enumerate}

        \item With every addition of a new point, $\mball{S_i}$ grows by a factor of at least $1 + \frac{\eps^2}{16}$.
        
        \item For any index $i$ it holds that $\ball{c^*_i}{r^*_i} \subset B_i$, at the latest when $|S_i| \geq \frac{32(1+\eps)}{\eps^3}$.
    \end{enumerate}
    
\end{lemma}

\begin{proof}
    \begin{enumerate}
        \item  We construct $u$ by recording the proper assignments in \textsc{Selection}. This is possible because the algorithm is deterministic and because assignments do not have to be specified before points have been selected. During the first iteration, if $p_1 \in \ball{c^*_{i_1}}{r^*_{i_1}}$, set $u_1 = i_1$, and so on. Note that the covering is disjoint, so this assignment is unambiguous. Obviously, $S_i \subseteq \ball{c^*_i}{r^*_i}$ has to hold for all $i$ necessarily and \ref{item:subset} follows.
        To prove \ref{item:singleton}, assume that there exists a singleton $S_i = \{s_i\}$ in the current iteration. From the previous point, we know that $s_i \in \ball{c^*_i}{r^*_i}$ and so
        $\ball{s_i}{\frac{\eps}{1 + \eps}\widetilde{r_i}} \subset\ball{s_i}{\eps r^*_i} \subseteq \ball{c^*_i}{(1 + \eps) r^*_i}.$
        Since $\mathscr{B}$ is at least $(1 + \eps)$-separated, this proves the second point.
        To prove \ref{item:balls-disjoint}, we have to look at how each $B_i$ is constructed. We start with a $(1 + \eps)$-approximation $\ball{c_i}{r_i}$ of $\mball{S_i}$. By definition, $\ball{c_i}{r_i} \subseteq \ball{c^*_i}{(1 + \eps)r^*_i}$ and setting $B_i = \ball{c_i}{\gamma r_i}$ ensures that $B_i \subseteq \ball{c^*_i}{(1 + \eps)\gamma r^*_i}$. The claim thus again follows from the assumption that $\mathscr{B}$ is $(1 + \eps) \gamma$-separated.

        \item This part of the proof does not deviate significantly from B\u{a}doiu et al.~\cite{BadoiuHI02} and has been moved to the appendix. The main difference is that we are working with $(1 + \eps)$-approximations of radii, which adds another layer of complexity.
        
        \item Each ball $B_i$ is a $(1 + \eps) \gamma$-approximation of $\mball{S_i}$. Since $S_i \subset \ball{c^*_i}{r^*_i}$ in each iteration, it follows that $B_i \subset \ball{c^*_i}{(1 + \eps)\gamma r^*_i}$ also holds throughout. No other ball (neither from \ref{item:singleton} nor \ref{item:balls-disjoint}) can intersect $\ball{c^*_i}{r^*_i}$ and so, by continually selecting new points, at some point $\ball{c^*_i}{r^*_i} \subset B_i$ must hold. We now want to show that this happens relatively quickly and that it is necessary to add at most $\frac{32(1+\eps)}{\eps^3}$ points to $S_i$ until this state is reached. Of course, it may happen that $\ball{c^*_i}{r^*_i}$ is covered at an earlier point in time and that fewer points have to be added to $S_i$. Assume that we get at least to an iteration, where $S_i$ contains two points. 
        Since we ignored all points that were at a distance of at most $\frac{\eps }{(1+\eps)}\widetilde{r_i} \leq \eps r^*_i$ from the first selected point, $\mball{S_i}$ has to have an initial radius of at least $\frac{\eps}{2(1+\eps)} \widetilde{r_i} \geq \frac{\eps r^*_i}{2(1 + \eps)}$. As we saw in point (2), any subsequent additions of new points to $S_i$ further increase the radius by a multiplicative factor of at least $ (1+\frac{\eps^2}{16})$. Combining both of these observations gives us an upper bound on the number of iterations. First, note that the initial radius of $\frac{\eps r^*_i}{2(1+\eps)}$ grows by at least $\frac{\eps^2}{16}  \cdot \frac{\eps r^*_i}{2(1+\eps)} = \frac{\eps^3 r^*_i}{32(1+\eps)}$ when the next point is added. Since the radii only grow larger, each subsequent update also increases the radius by at least this amount. At the same time, $r^*_i$ is clearly an upper bound for the radius of $\mball{S_i}$, so we can add at most $\frac{32(1+\eps)}{\eps^3}$ many points to $S_i$.
        \qedhere
    \end{enumerate}
\end{proof}

This lemma shows that we can reconstruct well-separated coverings (or rather, the corresponding clusterings) using a reasonably small oracle for the assignments, given that we know the radii up to an $\eps$-factor.

\subsection{Guessing the Radii}\label{sec:guessingradii}
Let us now consider how we can compute such approximate radii. We split this problem into two parts. First, we guess the largest radius of the covering and in the next step, we guess the remaining radii, assuming that they cannot be too small compared to this largest radius.

There are two different approaches to guessing the largest radius. The first one makes use of a relation between the largest radius in an optimal \textsf{$k$-MSR} solution and the value of an optimal $k$-center solution. If we have access to a constant-factor approximation algorithm for $k$-center under the given constraint, we can use it to compute a candidate set of small size for the largest radius.
The second approach uses results from the theory of $\eps$-coresets and works for arbitrary mergeable constraints with the trade-off that the set of candidates is larger and, in the end, depends exponentially on $k$. We focus on the first approach but refer to the appendix for the details and for the second approach.
The following lemma establishes a useful connection between $k$-center and \textsf{$k$-MSR}.
\begin{restatable}{lemma}{centervsradii}
    \label{k-center-vs-k-min-sum-radii}
    Let $r_\alpha$ denote the value of an $\alpha$-approximate $k$-center solution and $r_1^*$ the largest radius of a $\beta$-approximative \textsf{$k$-MSR} solution for the same instance. Then it holds that
        $r_{1}^*\in \left[\frac{r_\alpha}{\alpha}, \beta \cdot k^2 \cdot r_\alpha \right]$,
    even if we impose the same clustering constraints on both problems.
\end{restatable}

This means that, by running a (polynomial time) constant-factor approximation algorithm for $k$-center, we can obtain an interval $I$ which necessarily contains the radius $r_1^*$ of the \emph{largest} cluster in an optimal min-sum-radii solution. 
By utilizing standard discretization techniques, we are then able to obtain a finite candidate set such that $(a)$ its size only depends on $\eps$, $k$, $\alpha$ and $\beta$, and $(b)$ it contains a $(1+\eps)$-approximation for each value in $I$. The details can be found in the appendix.

Once we have a guess for $r_1^*$, we can apply a similar technique to obtain a candidate set for the remaining radii. However, this requires that the other radii are not too small in comparison. More precisely, we assume that the covering we are interested in is $\eps$-balanced. Later on, we will show that this requirement can easily be met.

\begin{definition}\label{def:eps-balanced}
    Let $\eps > 0$. A covering $\{\ball{c_1}{r_1} ,\ldots, \ball{c_k}{r_k}\}$ of $P$ is \emph{\ebalanced}, if $r_i \geq \frac{\eps}{k} \max\limits_j r_j$ for all $i \in \{1, \ldots, k\}$.
\end{definition}

Given such an $\eps$-balanced covering, we can conclude this part with the following statement, whose proof can also be found in the appendix.

\begin{restatable}{lemma}{guessingradii}
    \label{lem:approximate-radii}
    Let $\eps > 0$ and let $\mathscr{B}^*$ be an \ebalanced covering with radii $r_1^*,\ldots, r_k^*$. Then we can compute a set of size $O(\log_{(1+\eps)}k)$ that contains a number $r_1$ with $r_1^*\leq r_1 \leq (1+\eps)r_1^*$, and a set of size $O(\log_{(1+\eps)}\frac{k}{\eps})$
     that contains for each $r_i^*$, $i\geq 2$, a number $r_i$ with $r_i^* \leq r_i \leq (1+\eps)r_i^*$. 
\end{restatable}

\subsection{Cheap, Separable and Balanced Coverings}
\label{sec:nicelystructuredsolutions}

In the main technical lemma (\Cref{lem:construct-u}), we have proven that \textsc{Selection} is able to reconstruct well-separated coverings (clusterings), given approximate values for the radii. How those latter approximations could be computed was then outlined in~\Cref{sec:guessingradii}. What now remains to be shown is that there actually exist cheap, well-separated and balanced coverings, so that these results can be applied. 

\begin{lemma}
    \label{lem:existence_epsilon_separated}
    Let $\mathscr{C} = \{C_1, \ldots, C_k\}$ be a min-sum-radii solution for $P$. Then for all $\eps > 0$ and $\gamma \geq 1$, there exists an $\eps$-balanced and $\gamma$-separated covering $\mathscr{B} = \{B_1, \ldots, B_{k'}\}$ of $P$ with $k' \leq k$ and $\cost(\mathscr{B}) \leq (1 + \eps)^k \gamma^{k-1} \cost(\mathscr{C})$. 
    Additionally, if $\mathscr{C}$ satisfies a given mergeable constraint, then so does the corresponding clustering $\{B_1 \cap P, \ldots, B_{k'} \cap P\}$.
\end{lemma}

\begin{proof}
    Starting with $\{\mball{C_1}, \ldots, \mball{C_k}\}$, we construct $\mathscr{B}$ from $\mathscr{C}$ in phases consisting of two steps: (1) merge balls that are currently too close to each other and thus not $\gamma$-separated, (2) ensure that the current covering is \ebalanced by increasing the radii of balls that are too small. The first step increases the cost by a multiplicative factor of at most $\gamma$ and the second by a multiplicative factor of at most $(1 + \eps)$. Both steps are alternatively applied in phases until the resulting clustering is both \ebalanced and $\gamma$-separated. Now, even though step (1) might yield a covering that is neither \ebalanced nor $\gamma$-separated and step (2) might yield a clustering that is not $\gamma$-separated, since the number of balls reduces with every phase, except maybe the first, there can only be $k$ phases altogether. At that point, only one ball would remain and the clustering necessarily has to satisfy both properties. The resulting covering will cost at most $(1 + \eps)^k \gamma^{k-1} \cost(\mathscr{C})$.

    Let $B_1 = \ball{c_1}{r_1}, \ldots, B_{k'} = \ball{c_{k'}}{r_{k'}}$ denote the covering that has been constructed up to this point. For step (1), construct a graph $G$ on top of $\mathscr{B}$, where two balls $B_i$ and $B_j$ are connected, iff $\|c_i - c_j\| \leq \gamma (r_i + r_j)$. In other words, two balls are connected by an edge, if and only if they are not $\gamma$-separated.
    We try to construct a $\gamma$-separated covering by merging all balls that belong to the same connected component.
    This just means that we replace all balls of the connected component with the minimal-enclosing-ball of the connected component.
    Take any connected component $Z$ of $G$ and consider two arbitrary points $x, y \in \bigcup_{B_\lambda \in Z} B_\lambda$, say $x \in B_i$ and $y \in B_j$. 
    We can upper-bound the distance between them as follows: For any path, $B_i = B_{i_0}, \ldots B_{i_\ell} = B_j$ in $G$ that connects $B_i$ and $B_j$ we have 
    \begin{align*}
        \|x - y\| &\leq \|x - c_i\| + \sum_{\lambda = 0}^{\ell - 1} \|c_{i_\lambda} - c_{i_{\lambda + 1}}\| + \|y - c_j\| \\
        &\leq r_i + r_j + \sum_{\lambda = 0}^{\ell - 1} \gamma(r_{i_\lambda} + r_{i_{\lambda + 1}}) \leq \gamma \sum_{B_\lambda \in Z} 2r_\lambda
    \end{align*} 
    In other words, the radius of the resulting ball is larger than the sum of the previous radii by a factor of at most $\gamma$. At this point, we might end up in a situation similar to the one with which we started; there might again be balls that are too close together and thus not $\gamma$-separate. However, we have reduced the number of balls by at least one and so this step can be performed at most $k-1$ times.
    
    For step (2) let $r_{i_1}, \ldots, r_{i_\ell}$ denote all radii with $r_{i_j} < \frac{\eps}{k} \max_i r_i$. If we just set $r_{i_j} = \frac{\eps}{k} \max_i r_i$ for all $j \in \{1, \ldots, \ell\}$, this increases the cost of the covering by at most $\ell \frac{\eps}{k} \max_i r_i \leq \eps \max_i r_i \leq \eps \cost(\mathscr{C}')$. The resulting covering is necessarily \ebalanced. If it is not $\gamma$-separated we add another phase, starting with step (1).
\end{proof}

\begin{algorithm}[h]
    \caption{\textsc{Clustering}}
    \label{alg:clustering}
    \setstretch{1.2}
    \LinesNumbered
    \SetKwInOut{Input}{Input}
    \SetKwInOut{Output}{Output}
    \BlankLine
    \Input{An ordered set $P \subseteq \mathbb{R}^d$, a number $k \in \mathbb{N}$, a value $0 < \eps < 1$, a $(1 + \eps)$-approximation $r_\text{max}$ for largest radius}
    \Output{A $(1 + \eps)$-approximative $k$-clustering $\mathscr{C}$ of $P$}
    \BlankLine
        $\mathscr{C} \gets \{P, \emptyset, \ldots, \emptyset\}$\tcc*[r]{A feasible clustering to start with}\label{line:initialise}
        \ForAll{$(r_2, \ldots, r_{k}) \in \{(1+\eps)^{i}\frac{\eps}{k}r_\text{max} \; \mid \; i \in \{0, \ldots, \ceil{\log_{1+\eps}(\frac{k}{\eps})}\}\}^{k-1}$}{
            \ForAll{$u \in \{1, \ldots, k\}^{\frac{32k(1+\eps)}{\eps^3}}$\label{line:invoke}}{ 
                $B_1, \ldots, B_k \gets \textsc{Selection}(P, k, u, \{r_\text{max}, r_2, \ldots, r_k\}, \eps)$\;
                $\mathscr{C}' \gets C_1, \ldots , C_k$, where $C_i = B_i \cap P$\label{line:set-clustering}\;
                \If{$\mathscr{C}'$ is a valid clustering and $\cost(\mathscr{C}') < \cost(\mathscr{C})$}{
                    $\mathscr{C} \gets \mathscr{C}'$\;
                }
            }
        }
        \Return $\mathscr{C}$\;
\end{algorithm}

\subsection{The Main Result}

Now we are ready to prove the main theorem of this paper. 

\maintheorem*
\begin{proof}
    Set $\eps' = \left(\frac{\eps}{12k}\right)^2$, $\gamma = (1 + \epsp + 2\sqrt{\epsp})$ and let $\mathscr{C}^{\opt}$ be an optimal min-sum-radii solution that satisfies the mergeable constraint.   \Cref{lem:existence_epsilon_separated} shows that there exists a $(1 + \epsp)\gamma$-separated and $\epsp$-balanced covering $\mathscr{B}^* = \{\ball{c^*_1}{r^*_1}, \ldots \ball{c^*_k}{r^*_k}\}$ with $\cost(\mathscr{B}^*) \leq (1 + \epsp)^k \gamma^{k-1} \cost(\mathscr{C}^{\text{opt}})$. Denote the corresponding clustering by $\mathscr{C}^* = \{\ball{c^*_1}{r^*_1} \cap P, \ldots, \ball{c^*_1}{r^*_1} \cap P\}$ and assume that the balls are ordered such that $r^*_1$ is the largest radius.
    Using Algorithm~\ref{lem:approximate-radii}, we can compute approximate radii $\widetilde{r_1}, \ldots, \widetilde{r_k}$, such that $\cost(C^*_i) \leq \widetilde{r_i} \leq (1 + \epsp)\cost(C^*_i)$ for all $i$. Consider now, for $u^*$ as in \Cref{lem:construct-u}, the variables at the end of $\textsc{Selection}(P, k, u^*, \{\widetilde{r_1}, \ldots, \widetilde{r_k}\}, \eps')$. Since none of the $B_i$ overlap and $C^*_i \subseteq B_i$ for all $i$, Algorithm~\ref{alg:selection} is able to fully reconstruct $\mathscr{C}^*$. Additionally, since the length of $u^*$ does not exceed $k \frac{32(1 + \eps')}{\eps'^3}$, this necessarily happens in one of the iterations of $\textsc{Clustering}(P, k, \eps', \widetilde{r_1})$.
    As such, running Algorithm~\ref{alg:clustering} for all possible guesses of maximal radii provided by \Cref{lem:guess-largest-radius} guarantees an approximation ratio of $(1+\eps')^k\gamma^{k-1} \leq (1+3\sqrt{\eps'})^{2k}.$ Substituting $\eps' = \left( \frac{\eps}{12k}\right)^2$, we get an approximation ratio of $(1 + 3 \sqrt{\eps'})^{2k} \leq (1+\frac{\eps}{4k})^{2k} \le e^{\eps/2} \le 1 + \eps$ for $\eps\le 1/2$. 

    For the running time, we start by analyzing the time needed for one call to \textsc{Selection}. Initializing the $S_i$ takes $O(k)$ and so does the final loop in line $14$. The main for-loop iterates over $u$, which has length $32k(1+\eps')/\epsp^3$. 
    The computation of the $(1+\epsp)$-approximation of the minimum enclosing balls in line $9$ can be done in $O\left(|S_{u_i}|\cdot d/\eps' \right) = O\left(d\cdot\poly(k, 1/\epsp)\right)$ with the algorithm from \cite{Yildirin2008}. 
    Thus, a single call to \textsc{Selection} takes $O\left( d \cdot \poly(k,1/\epsp)\right)$.

    In \textsc{Clustering}, we have two nested for-loops, that go through $O\left((\log_{1+\eps'} (k/\eps'))^{k-1}\right)$ and $k^{\frac{32k(1+\eps')}{\eps'^3}}$ iterations respectively, in each of which \textsc{Selection} is invoked. 
    Line $6$, which checks whether the clustering covers the whole set and satisfies the constraint, takes at most $O(\poly(n))$ time (depending on the constraint, this time might even be linear in $n$).
    Thus, one call to \textsc{Clustering} takes $d \cdot \poly(n) \cdot k^{O(\poly(k,1/\eps))} \cdot (\log_{1+\eps'} \poly(k,1/\eps))^{k-1}$ time.

    Finally, \textsc{Clustering} has to be called for every candidate for $r_{\text{max}}$. There are at most $O(k + \log_{1 + \eps} \gamma^{k - 1} k)$ such candidates and so the overall running time is $d \cdot \poly(n) \cdot f(k, \eps)$.
\end{proof}

Using the other method of guessing the largest radius extends the result to \emph{all} mergeable clustering constraints. The trade-off is a worse running time.

\begin{theorem}
    \label{thm:result_2}
    For every  $0< \eps < 1/2$, there exists an algorithm that computes a $(1+\eps)$-approximation for min-sum-radii with mergeable constraints in time $d \cdot n^{\poly(k, 1/\eps)} \cdot f(k, \eps).$
\end{theorem}

\begin{proof}
    We follow the arguments in \Cref{thm:result}, with the only difference being the computation of the candidate set $R$ for the largest radius. \Cref{lem:approximate-largest-radius} (appendix) implies that we can compute a candidate set of size $n^{O(1/\eps')}$ that contains a $(1+\eps')$-approximation for the largest radius in time $\frac{d}{\eps'^2}n^{O(1/\eps')}$. Substituting $\eps' = \left( \frac{\eps}{12k}\right)^2$, we get the purported runtime.
\end{proof}

\subparagraph*{A Word on Outliers.} We have stated the allowance of outliers as a mergeable constraint in the introduction. Note, however, that there is one issue: When we want to achieve $\eps$-separation, we may not start to merge outliers into clusters with more than one point. So clustering with outliers is not strictly mergeable. However, the algorithm can be suitably adapted: Only make sure that the non-outlier clusters are separated, and during the oracle calls, 
only let the oracle decide whether a point is an outlier or not, and if not, to which cluster it belongs. We do not derive the details of such an algorithm in this paper.

\newpage
\bibliography{reference}

\appendix

\newpage
\section{Omitted Proofs}

Proof of point 2 of~\Cref{lem:construct-u}.

\begin{proof}
        To show that the radius increases by a multiplicative factor of $(1 + \frac{\eps^2}{16})$, consider the addition of a new point $p \in P$ to $S_i$. Denote the set $S_i$ immediately before $p$ is added by $S_i^{\text{old}}$ and by $S_i^{\text{new}}$ directly afterwards. Let $c_i^{\text{old}}$, $r_i^{\text{old}}$ be the center and radius of $\mball{S_i^{\text{old}}}$ respectively and $c_i^{\text{new}}$, $r_i^{\text{new}}$ be those of $\mball{S_i^{\text{new}}}$. Additionally, let $B_i = \ball{c_i}{r_i}$ denote the $(1 + \eps)$-approximation for $\mball{S_i^{\text{old}}}$ computed in the algorithm.
        
        First, we show that $c_i$ and $c_i^{\text{old}}$ cannot be too far apart. Draw a line $L$ through $c_i^{\old}$ and $c_i$ and let $H$ be the $(d-1)$-dimensional hyperplane orthogonal to $L$ anchored $c_i^{\text{old}}$. We denote the open halfspace induced by $H$ that does not contain $c_i$ by $H^-$. Lemma 2.2 from~\cite{BadoiuHI02} shows that there exists a point $x$ in $H^- \cap S_i^{\old}$ that is at a distance exactly $r_i^{\old}$ from $c_i^{\old}$. Since the triangle drawn by $x$, $c_i^{\old}$ and $c_i$ has an obtuse angle at $c_i^{\old}$, we can apply the law of cosines to get 
        \[ (1 + \eps)^2 \left(r_i^{\text{old}}\right)^2 \geq r_i^2 \geq \|x - c_i\|^2 \geq \|c_i^{\text{old}} - x\|^2 + \|c_i - c_i^{\text{old}}\|^2 \geq \left(r_i^{\text{old}}\right)^2 + \|c_i - c_i^{\text{old}}\|^2, \]
        so that $\|c_i - c_i^{\text{old}}\| \leq \sqrt{\eps (2 + \eps)} r_i^{\text{old}} \leq \sqrt{3 \eps}r_i^{\text{old}} \leq 2\sqrt{\eps}r_i^{\text{old}}$. 
        
        We now make the following case distinction. First, if $\|c_i^{\old} - c_i^{\text{new}}\| < \frac{\eps}{2}r_i^{\text{old}}$, then using the triangle inequality twice yields
        \begin{align*}
             r_i^{\text{new}} &\geq 
             \|p - c_i^{\text{new}}\|
             \geq \|p - c_i\| - \|c_i - c_i^{\old}\| - \|c_i^{\old} - c_i^{\text{new}}\| \\
             &\geq \gamma r_i - 2\sqrt{\eps} r_i^{\old} - \frac{\eps}{2}r_i^{\old} \geq (1 + \frac{\eps}{2})r_i^{\text{old}} \geq (1+\frac{\eps^2}{16})r_i^{\text{old}}
        \end{align*}
        and we are done.
    
        Next, consider the case where $\|c_i^{\old} - c_i^{\text{new}}\| \geq \frac{\eps}{2}r_i^{\text{old}}$. The same argument via halfspaces as before shows that there exists a point $x \in S_i^{\text{old}}$ with $\|x - c_i^{\text{old}}\| = r_i^{\text{old}}$, such that the triangle drawn by $x$, $c_i^{\text{old}}$ and $c_i^{\text{new}}$ has an obtuse angle at $c_i^{\text{old}}$. Again, applying the law of cosines yields
        \begin{align*}
            r_i^{\text{new}} 
            &\geq \left\|c_i^{\text{new}} - x\right\| 
            \geq \sqrt{\|x - c_i^{\old}\|^2+ \|c_i^{\old} - c_i^{\text{new}}\|^2}\\
            &\geq \sqrt{\left(r_i^{\old}\right)^2 + \frac{\eps^2}{4}\left(r_i^{\old}\right)^2} 
            \geq \left(1 + \frac{\eps^2}{16}\right)r_i^{\old} 
        \end{align*} 
        for $0 < \eps < 1$.
\end{proof}

\subsection{Guessing the Radii}
In this section, we give the proofs for the statements in \Cref{sec:guessingradii}. We start with the relationship between $k$-center and \textsf{$k$-MSR}.
\centervsradii*
\begin{proof}
    Let $r_c$ denote the value of an \emph{optimal} $k$-center solution.
    Notice that any $k$-center solution is also a feasible \textsf{$k$-MSR} solution and vice versa.
    As $r_c$ is the radius of a largest cluster in the optimal $k$-center solution, we have $r^*_1\geq r_c$, as otherwise the \textsf{$k$-MSR} solution would be a better solution for $k$-center. Since $r_\alpha$ is an $\alpha$-approximation for $r_c$ (i.e. $r_\alpha \leq \alpha\cdot r_c$), we get 
    $r_1^*\geq r_c \geq r_\alpha/\alpha.$
    
    On the other hand, the sum of the radii in the $k$-center solution is at most $k\cdot r_\alpha$. Let $r_{\opt}$ be the largest radius in an optimal \textsf{$k$-MSR} solution. It follows that $r_{\opt} \leq k\cdot r_{\alpha}$, as otherwise, the $k$-center solution would yield a better \textsf{$k$-MSR} objective value than the optimal \textsf{$k$-MSR} solution.
    Therefore we must have $r_1^*\leq \beta \cdot k^2 \cdot r_\alpha$, since otherwise we would have $r_1^* > \beta \cdot k^2 \cdot r_\alpha \geq \beta \cdot k \cdot r_{\opt}$, contradicting the assumption that $r_1^*$ is the largest radius in a $\beta$-approximation for \textsf{$k$-MSR}.
\end{proof}
To prove \Cref{lem:approximate-radii}, we start with the following useful lemma. It encapsulates the common technique of ``covering'' an interval $[a,b]$ by $O(\log_{(1+\eps)}(b/a))$ smaller intervals to obtain a (somewhat reasonably sized) discrete set of values, that contains for each value in $[a,b]$ a $(1+\eps)$-approximation.

\begin{lemma}\label{cover-interval-with-eps-powers}
    Let $\eps > 0$,  $[a,b] \subset \mathbb{R}_{\geq 0}$ and $m = \ceil{\log_{(1+\eps)}(b/a)}$. Then for every $r^* \in [a,b]$, the set
    $R = \{(1+\eps)^{i}a\; \mid \; i \in \{0, \ldots, m\}\}$
    contains a number $r$ with $r^* \leq r \leq (1+\eps)r^*$.
\end{lemma}
\begin{proof}
    For $j\in \{1,\ldots m\}$, let $I_j = \left[(1+\eps)^{j-1}a, (1+\eps)^ja\right]$. Note that $R$ consists of the endpoints of these intervals and that $\bigcup\limits_{j=1}^{m}I_j \supseteq [a,b]$. Thus, for every $r^*\in [a,b]$, there exists some $i \in \{1,\ldots,m\}$ such that $r^*\in I_i$. Then for $r =  (1+\eps)^ia$ (i.e. the right endpoint of $I_i$) we clearly have $r^* \leq r \leq (1+\eps)r^*$, which proves the claim. 
\end{proof}

This can directly be applied to the interval given by \Cref{k-center-vs-k-min-sum-radii} to obtain a candidate set for the largest radius in any feasible solution whose size does not depend on $n$.

\begin{lemma}
    \label{lem:guess-largest-radius}
    Let $\mathscr{B}^* = \{B_1, \ldots, B_k\}$ be a $\gamma$-separated covering with largest radius $r_1^*$, whose corresponding clustering $\{B_1 \cap P, \ldots, B_k \cap P\}$ is an $\beta$-approximation for \textsf{$k$-MSR}. If $r_\alpha$ denotes the value of an $\alpha$-approximate $k$-center solution, then the set 
    \[R = \left\{(1+\eps)^{i}\frac{r_\alpha}{\alpha}\; \mid \; i \in \left\{0, \ldots, 2\ceil{\log_{(1+\eps)}\alpha \beta k}\right\}\right\}\]
    contains a number $r_1$ with $r_1^* \leq r_1 \leq (1+\eps)r^*_1$.
\end{lemma}

\begin{proof}
    By \Cref{k-center-vs-k-min-sum-radii}, the interval $I = \left[\frac{r_\alpha}{\alpha}, \beta\cdot k^2\cdot r_\alpha \right]$ contains $r_1^{*}$. We now cover this interval with smaller intervals. For $j \in \left\{1, \ldots, \ceil{\log_{(1+\eps)}\alpha \beta k^2}\right\}$, let $I_j = \left[(1+\eps)^{j-1}\frac{r_\alpha}{\alpha}, (1+\eps)^{j}\frac{r_\alpha} {\alpha}\right]$. Notice that $\bigcup_{j}I_j \supseteq I$, so for each candidate value $r_1\in I$ there is some $i$ such that $r_1\in I_i$. Since the set $R$ consists of the endpoints of these intervals, the corresponding endpoint $r=(1+\eps)^i\frac{r_\alpha}{\alpha}$ fulfills $r_1\leq r \leq (1+\eps)r_1$, proving the claim.
\end{proof}

Once the largest radius is fixed, we can obtain a candidate set for all other radii in a similar manner. The difference here is that we need a suitable lower bound for the radii, so that the candidate set does not get too large. This is achieved by assuming that the solution we are trying to approximate is \ebalanced, as introduced in \Cref{def:eps-balanced}. As shown in \Cref{lem:existence_epsilon_separated}, this is not too big of a restriction.

\begin{lemma}
\label{lem:remaining-radii}
Let $\eps > 0$ and let $\mathscr{B}^*$ be an \ebalanced covering with radii $r_1^*,\ldots, r_k^*$. Then the set \[R = \left\{(1+\eps)^j \frac{\eps r_1^*}{k}~|~i\in \{0,\ldots,\frac{k}{\eps}\}\right\}\]
contains for each $r_i^*$ a number $r_i$ with $r_i^* \leq r_i \leq (1+\eps)r_i^*$.
\end{lemma}
\begin{proof}
    Since the solution is \ebalanced, we have $r_i^* \in \left[\frac{\eps r_1^*}{k}, r_1^*\right]$ for all $i\in \{2,\ldots, k\}$. The claim then immediately follows from \Cref{cover-interval-with-eps-powers}.
\end{proof}
This now immediately implies \Cref{lem:approximate-radii}.
\guessingradii*
\begin{proof}
    \Cref{lem:guess-largest-radius} shows how to obtain the candidate set for the largest radius of the desired size, \Cref{lem:remaining-radii} shows it for the remaining radii.
\end{proof}
As mentioned previously, there is an alternative way of guessing the radii that is not require an approximation algorithm for the $k$-center variant as a subroutine. The drawback of this approach is that the size of the candidate set now depends on $n$.
\begin{lemma}
    \label{lem:approximate-largest-radius}
    There exists a set $R \subset \mathbb{R}$ of size $n^{O(\frac{1}{\eps})}$ computable in $\frac{d}{\eps^2} n^{O(\frac{1}{\eps})}$ time, that contains, for all sets $C \subset P$, a radius $r_C \in R$, such that $\cost(C) \leq r_C \leq (1 + \eps) \cost(C)$.
\end{lemma}
\begin{proof}
    We combine two results from Yıldırım  the theory on $\eps$-coresets. An $\eps$-coreset $K$ of a set $X \subset \mathbb{R}^d$ is a small subset which is almost as large as $X$, in the sense that $\cost(X) \leq (1+\eps)\cost(K)$. It is known, that for every subset $C \subset \mathbb{R}^d$ there exists an $\eps$-coreset of size at most $O(\frac{1}{\eps})$ (cf. Theorem 3.3 of~\cite{Yildirin2008}). This gives us our search space. For all subsets $K \subset P$ of size at most $O(\frac{1}{\eps})$ we compute a $(1 + \eps)$-approximation $B_K = \ball{c_K}{r_K}$ of $\mball{K}$ using Algorithm~3.1 of~\cite{Yildirin2008}. Since $K$ is an $\eps$-coreset and $B_K$ a $(1 + \eps)$-approximation we can do the following case distinction. If $r_K \leq \cost(C)$, then
    \[ \cost(C) \leq (1 + \eps)\cost(K) \leq (1 + \eps)r_K \leq (1 + \eps) \cost(C)\]
    and if $r_K > \cost(C)$, then
    \[ \cost(C) \leq r_K \leq (1 + \eps) \cost(K) \leq (1 + \eps) \cost(C). \]
    One of the two, $r_K$ or $(1+\eps)r_K$, must satisfy the condition posed in the lemma, so we can just add both to $R$.
    There exist at most $n^{O\left(\frac{1}{\eps}\right)}$ different subsets $K \subset P$ of size at most $O(1/\eps)$, so it follows that $R$ is at most twice as large. Combining this with the fact that Yildirim's algorithm has a running time of $O(d/\eps^2)$ on these small instances (see Theorem 3.2 of~\cite{Yildirin2008}) completes the proof.
\end{proof}
So even if we want to solve \textsf{$k$-MSR} with constraints for which we do not have access to an approximation algorithm for the corresponding $k$-center problem, we can employ this method, which leads to the overall runtime shown in \Cref{thm:result_2}. 

\section{Mergeable Constraints}\label{sec:mergeable}

Let $\mathscr{C} = \{C_1, \ldots, C_k\}$ be a clustering of a finite set of points $P$.
Recall the definition of {\em mergeable constraints}:

\setcounter{definition}{0}
\begin{definition}
    A clustering constraint is \emph{mergeable} if the union $C \cup C'$ of any possible pair of clusters $C, C'$ satisfying the constraint does itself satisfy the constraint (cf.~\cite{Arutyunova021}). In other words, merging clusters does not destroy their property of satisfying the constraint.
\end{definition}

In the following, we will list a few mergeable constraints and explain why they have this property.

\setcounter{definition}{13}
\begin{definition}[Uniform Lower Bounds]
    Let $l \in \mathbb{N}$.
    $\mathscr{C}$ fulfills the uniform lower bounds constraint if for $|C_i| \geq l$ for every $i\leq k$. 
\end{definition}
The uniform lower bounds constraint obviously is mergeable as the union of two clusters contains at least as many points as any of the two individually.

\begin{definition}[Outliers]
    Let $z\in \mathbb{N}$. We say that $\mathscr{C}$ is a clustering of $P$ with at most $z$ outliers if $|\cup_{i\leq k} C_i| \geq |P|-z$.
\end{definition}
Let $\mathscr{C}$ be a clustering fulfilling the outlier constraint with $z\in \mathbb{N}$. Then, by merging two clusters, the number of points covered by the resulting clustering does not decrease. However, we may not merge two outlier clusters. Thus, the outlier constraint is kind-of mergeable -- one would need to adapt our algorithm a bit.

For defining the fairness notions, we need a finite set of colors $\mathcal{H}$ that constitute group memberships. Every point in the set $P$ is assigned a color that indicates which group this point belongs to.
For $X\subseteq P$ and $h\in \mathcal{H}$, let $\col_{h}(X) \subseteq X$ denote the subset of points within $X$ that carry color $h$.

The {\em exact fairness} notion from \cite{bercea2018cost} requires that in every cluster, the proportion of points of a certain color is the same as the proportion of this color within the complete point set $P$.
\begin{definition}[Exact Fairness]
    The clustering $\mathscr{C}$ fulfills {\em exact fairness} if 
    \[ \frac{|\col_{h}(C)|}{|C|} = \frac{|\col_h(P)|}{|P|} \]
    for every $h\in \mathcal{H}$ and $C\in \mathscr{C}$.
\end{definition}
Because of the easy observation $\frac{a}{b} = \frac{c}{d} \Rightarrow \frac{a+c}{b+d} = \frac{a}{b}$ for $a,b,c,d \in \mathbb{R}^+,\ b,d\neq 0$, it directly follows that this constraint is mergeable.
Bercea et al. \cite{bercea2018cost} provide a 5-approximation for $k$-center under the exact fairness constraint.

Among the fairness notions, there also exist several different definitions of what it means for a clustering to be {\em balanced}.
The first notion is defined for two colors.
\begin{definition}[Balance Notion by Chierichetti et al.~\cite{chierichetti2017fair}]
    Let $\mathcal{H} = \{h_1, h_2\}$ and $b\in [0,1]$. The clustering $\mathscr{C}$ is said to have a balance of at least $b$ if 
    \[ \min_{C\in\mathscr{C}} \min \left\{ 
    \frac{|\col_{h_1}(C)|}{|\col_{h_2}(C)|}, \frac{|\col_{h_2}(C)|}{|\col_{h_1}(C)|}  \right\} \geq b. \]
\end{definition}
It can easily be seen that the balance of two merged clusters is at least as high as the balance of any individual cluster. 
There exists a 14-approximation for $k$-center under this notion of fairness by Rösner and Schmidt \cite{rosner2018privacy}. For the special case that the balance is lower bounded by 1 or $1/t$ for some positive integer $t$, there are even a 3- and a 4-approximation, respectively \cite{chierichetti2017fair}.

The following balance notion proposed in \cite{bohm2020fair} is more general in that it allows for an arbitrary number of colors, but the form of the constraint itself is stricter.
\begin{definition}[Balance Notion by Böhm et al.~\cite{bohm2020fair}]
    We say $\mathscr{C}$ is {\em exactly balanced} if 
    \[|\col_{h_1}(C)| = |\col_{h_2}(C)|\]
    for all $C\in \mathscr{C}$ and all $h_1, h_2 \in \mathcal{H}$.
\end{definition}
If for all $C\in \mathscr{C}$, $|\col_{h_1}(C)| = |\col_{h_2}(C)|$, then for $C_1, C_2 \in \mathscr{C}$, it easily follows
\[ |\col_{h_1}(C_1 \cup C_2)| = |\col_{h_1}(C_1)| + |\col_{h_1}(C_2)| = |\col_{h_2}(C_1)| + |\col_{h_2}(C_2)| = |\col_{h_2}(C_1 \cup C_2)|. \]
Böhm et al. \cite{bohm2020fair} show how the $k$-center problem under this fairness notion can be reduced to the unconstrained variant while increasing the approximation factor by 2. 
A more relaxed notion is the one used in \cite{bera2019fair}, as it allows to specify a range for every color constraining the number of points of this color in any cluster.
\begin{definition}[Balance Notion by Bera et al.~\cite{bera2019fair}]
    For every $h\in \mathcal{H}$, let $\alpha_h, \beta_h \in [0,1]$. The clustering $\mathscr{C}$ is said to be balanced with respect to the vectors $(\alpha_h)_{h\in \mathcal{H}}$ and $(\beta_h)_{h\in \mathcal{H}}$ if
    \[ \alpha_h|C| \leq |\col_h(C)| \leq \beta_h|C| \]
    for all $C\in \mathscr{C}$ and $h\in \mathcal{H}$.
\end{definition}
If the clusters $C_1, C_2 \in \mathscr{C}$ fulfill the constraints from the definition above, then
\[ \alpha_h |C_1 \cup C_2| = \alpha_h|C_1| + \alpha_h|C_2| \leq |\col_h(C_1)| + |\col_h(C_2)| \leq \beta_h |C_1| + \beta_h|C_2| = \beta_h|C_1\cup C_2| \]
for every $h\in \mathcal{H}$, and because $C_1$ and $C_2$ as well as $\col_h(C_1)$ and $\col_h(C_2)$ are disjoint it is $|\col_h(C_1)| + |\col_h(C_2)| = |\col_h(C_1) \cup \col_h(C_2)| = |\col_h(C_1\cup C_2)| $ and hence the union $C_1\cup C_2$ satisfies the definition.
Harb and Shan \cite{harb2020kfc} give a 5-approximation for the $k$-center problem under this constraint. 

The special case that only considers upper bounds, i.e. all $\alpha_h$ are set to 0, is also sometimes referred to as {\em bounded representation} or {\em $l$-diversity} constraint \cite{ahmadian2019clustering}.

\end{document}

%% file: drawing.tex
\begin{figure}[t]
\centering
\begin{tikzpicture}
\def \pointsize {0.05cm}

\draw (-0.05,-0.05) -- (0.05,0.05);
\draw (0.05,-0.05) -- (-0.05,0.05);

\def \number {32}
\def \radius {1.5cm}
\def \degree {360/\number}

\foreach \s in {1,...,\number}{
 \node [circle,draw, fill=black,inner sep=0cm, minimum width=\pointsize] at ({\degree * (\s -1)}:\radius) (a\s) {};
}

\def \number {32}
\def \radius {2cm}
\def \degree {360/\number}

\foreach \s in {1,...,\number}{
 \coordinate  (c\s) at ({\degree * (\s -1)}:\radius);
}

\def \number {8}
\def \radius {0.25cm}
\def \degree {360/\number}

\foreach \F in {13,15,17,19,21,1,3,5,31,29}{
  \foreach \s in {1,...,\number}{
   \node [circle,draw,fill=black,inner sep=0cm, minimum width=\pointsize] at ($({\degree * (\s -1)}:\radius)+(c\F)$) (j\s) {};
}
}

\begin{scope}[xshift=8cm]
\def \number {16}
\def \radius {1.5cm}
\def \degree {360/\number}

\foreach \s in {1,...,\number}{
 \node [circle,draw,fill=black,inner sep=0cm, minimum width=\pointsize] at ({\degree * (\s -1)}:\radius) (g\s) {};
}

\node [fill=orange,draw=orange,circle,inner sep=0cm, minimum width=0.15cm, label=left:{$c_1$}] at ($(g9)-(0.5,0)$) {};
\node [fill=blue,draw=blue,circle,inner sep=0cm, minimum width=0.15cm, label=right:{$x$}] at ($(g9)$) {};

\draw (-0.05,-0.05) -- (0.05,0.05);
\draw (0.05,-0.05) -- (-0.05,0.05);
\node [label=above:{$c_2$}] at (0,0) {};

\draw ($(g9)+(-0.05,-0.05)-(0.5,0)$) -- ($(g9) + (0.05,0.05) - (0.5,0)$);
\draw ($(g9)+(0.05,-0.05)-(0.5,0)$) -- ($(g9) + (-0.05,0.05) - (0.5,0)$);

\end{scope}

\end{tikzpicture}
\caption{Left side: An example where $k$-min-sum-radii rather opens one cluster than eleven. Right side: An example where the cheapest $k=2$-clustering keeps $c_1$ as a singleton rather than combining it with $x$, despite the fact that $c_1$ is closer to $x$ than the center $c_2$ of the big cluster.
It is cheaper to \emph{not} assign the blue point to the orange point even though that would be a closer center.\label{fig:counterintuitive}}
\end{figure}
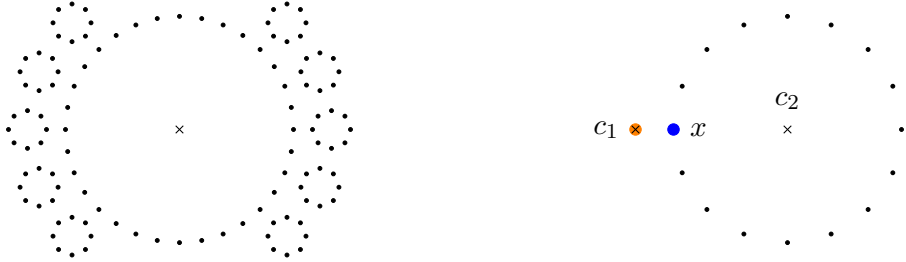

%% file: example-arxiv-pdfs.tex
\begin{figure}[t]
    \centering

    \begin{subfigure}[t]{0.32\textwidth}
    \centering
    \resizebox{\linewidth}{!}{
    \includegraphics{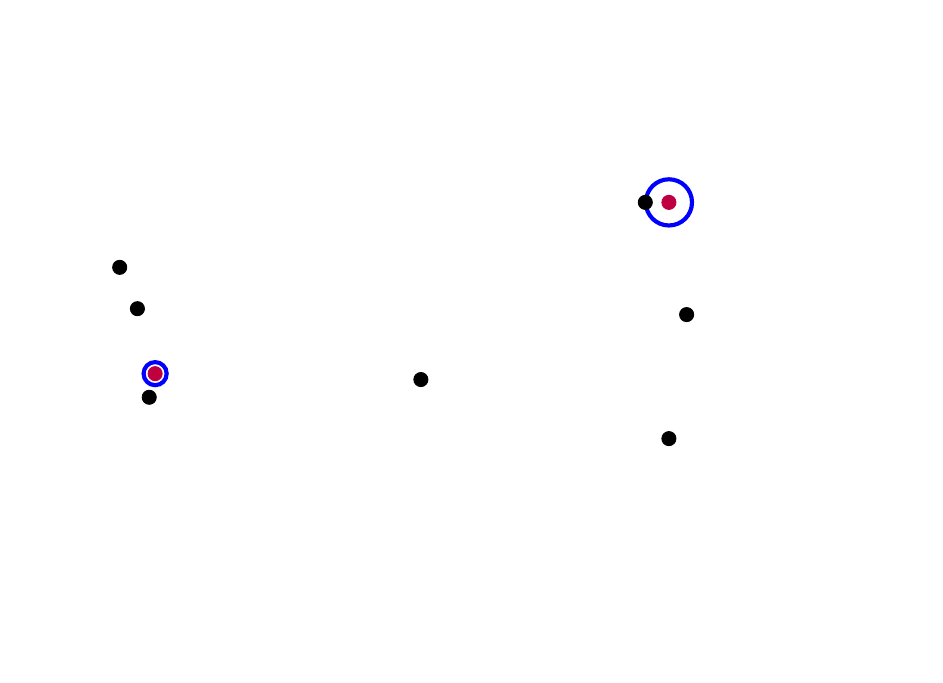}    }    \caption{The first two points are chosen and equipped with a small ball to start with.}
    \end{subfigure}\hfill
    \begin{subfigure}[t]{0.32\textwidth}
    \centering
    \resizebox{\linewidth}{!}{\includegraphics{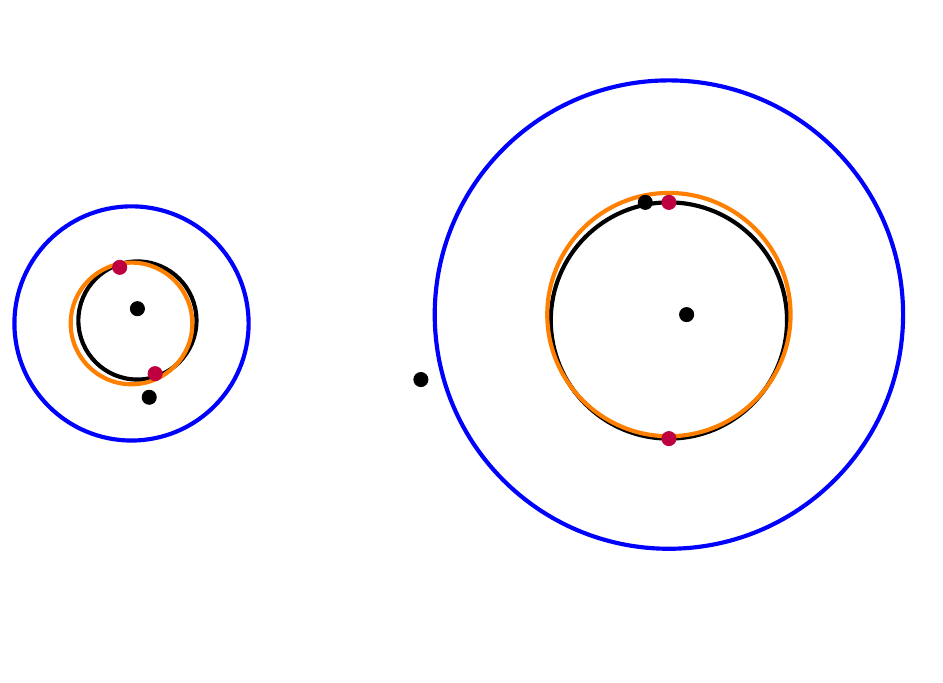}    }    \caption{Two iterations later, a second point has been found for both left and right cluster, so the small balls have been replaced.}
    \end{subfigure}\hfill
    \begin{subfigure}[t]{0.32\textwidth}
    \centering
    \resizebox{\linewidth}{!}{\includegraphics{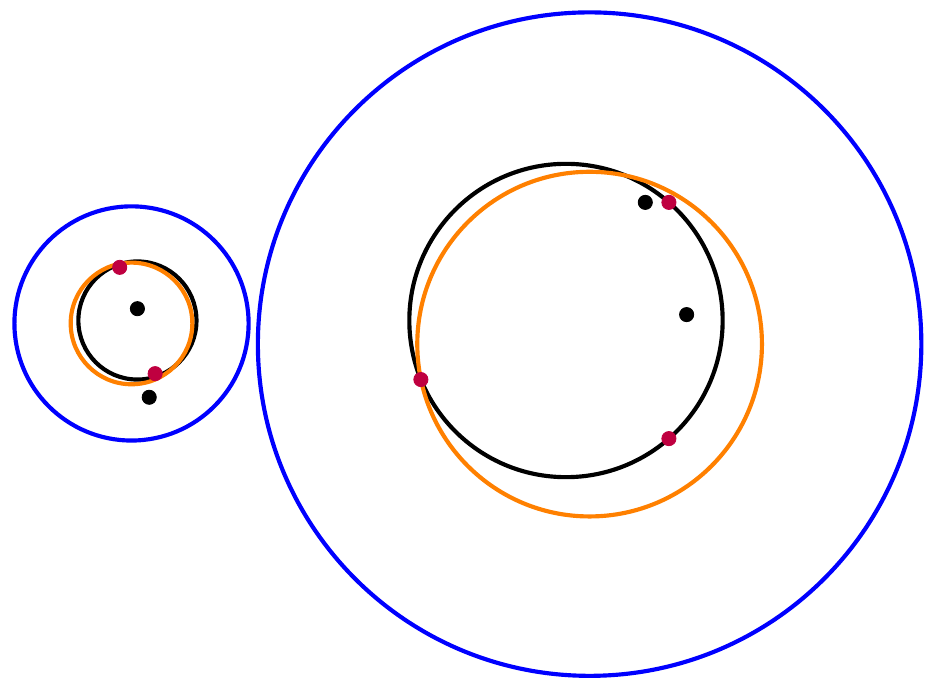}}
    \caption{After a third point was added to the right cluster, all points in the point set are covered and the run ends.}
    \end{subfigure}\hfill
    \caption{Example run of the algorithm for $\eps=0.2$. In every iteration, the purple points depict the points that were already chosen by the algorithm. For reference, the black circles represent the true minimum enclosing balls of the subsets $S_i$. These are not computed in the algorithm, but only the $(1+\eps)$-approximations of these, depicted in orange. 
    The blue circles enclose the areas which we ignore when sampling new points. That is, for singleton clusters, it is $B(s_j, \frac{\eps}{1+\eps}\widetilde{r_j})$, and once the algorithm found two points from a cluster, it computes an approximate MEB and enlarges it by a factor of $\gamma = 1+\eps + 2\sqrt{\eps}$ to obtain the new blue ball.}
    \label{fig:selection}
\end{figure}